\documentclass[12pt,a4paper]{article}
\usepackage[utf8]{inputenc}
\usepackage[english]{babel}
\usepackage{amsmath}
\usepackage{amsthm}
\usepackage{amsfonts}
\usepackage{amssymb}
\usepackage{geometry}
\usepackage[dvipdfmx]{graphicx}
\usepackage[longnamesfirst,authoryear]{natbib}
\usepackage{afterpage}
\usepackage{bm}
\usepackage{here}

\newtheorem{Theorem}{Theorem}
\newtheorem{Lemma}{Lemma}

\newtheorem{Assumption}{Assumption}

\newtheorem{Remark}{Remark}

\begin{document}

\title{Bandwidth Selection for Treatment Choice with Binary Outcomes\footnote{This work was partially supported by JSPS KAKENHI Grant Number 22K13373. I would like to thank the editor and an anonymous referee for their careful reading and comments. I would also like to thank Toru Kitagawa, Masayuki Sawada, and Kohei Yata for their helpful comments and suggestions.}}

\author{Takuya Ishihara\thanks{Graduate School of Economics and Management, Tohoku University, 41 Kawauchi, Aoba-ku, Sendai, Miyagi 980-0862, Japan. \textit{E-mail adress}: \texttt{takuya.ishihara.b7@tohoku.ac.jp}}}

\date{\today}
\maketitle

\begin{abstract}
This study considers the treatment choice problem when outcome variables are binary. We focus on statistical treatment rules that plug in fitted values based on nonparametric kernel regression and show that optimizing two parameters enables the calculation of the maximum regret. Using this result, we propose a novel bandwidth selection method based on the minimax regret criterion. Finally, we perform a numerical analysis to compare the optimal bandwidth choices for the binary and normally distributed outcomes. 
\end{abstract}

\section{Introduction}
This study examines the problem of determining whether to treat individuals based on observed covariates. A standard approach is to employ a plug-in rule that selects the treatment according to the sign of an estimate of the conditional average treatment effect (CATE). Kernel regression is a prevalent technique for estimating the CATE, and a crucial aspect of implementing this method is the decision regarding bandwidth selection. Many studies have proposed bandwidth selection approaches to solve the estimation problem (see \citet{li2007nonparametric}). However, to the best of our knowledge, there are few studies that investigate the bandwidth selection for the treatment choice problem. We propose a novel method for bandwidth selection in the treatment choice problem when dealing with binary outcome variables.

In this study, we consider the planner who want to determine whether to treat individuals with a particular covariate value based on experimental data. Following \citet{manski2004statistical, Manski2007}, \citet{stoye2009minimax, stoye2012minimax}, \citet{tetenov2012statistical}, \citet{ishihara2021evidence}, and \citet{yata2021optimal}, we focus on the minimax regret criterion to solve the decision problem. When the outcome variables are binary, the conditional distributions of the outcomes are characterized by conditional mean functions. We assume that the conditional mean functions are Lipschitz functions and show that the maximum regret can be calculated by optimizing two parameters. Based on these results, we propose a computationally tractable algorithm for obtaining the optimal bandwidth.

\citet{ishihara2021evidence} and \citet{yata2021optimal} derive the minimax regret rule in a similar setting when outcome variables are normally distributed. Using the argument of \citet{ishihara2021evidence}, the calculation of the maximum regret for a nonrandomized statistical treatment rule, which incorporates fitted values derived from nonparametric kernel regression, can be performed with ease. However, this approach relies on the normality and therefore can not be applied to binary outcomes.

\citet{stoye2012minimax} considers the statistical decision problems with binary outcomes. The setting of \citet{stoye2012minimax} is similar to that of this study, but our framework differs from two perspectives. First, we focus on the treatment choice at a particular covariate value, whereas \citet{stoye2012minimax} considers treatment assignment functions that map from the covariate support into a binary treatment status. Second, our restriction on the conditional mean functions differs from that of \citet{stoye2012minimax}. \citet{stoye2012minimax} assumes that the variations of the conditional mean functions are bounded. By contrast, our study considers the conditional mean functions as Lipschitz functions.

The remainder of this paper is organized as follows. Section 2 explains the study setting. Section 3 defines the statistical decision problem and provides a computationally tractable algorithm to obtain the optimal bandwidth. Section 4 presents a numerical analysis to compare bandwidth selections for binary and normally distributed outcomes. Section 5 concludes the study.

\section{Settings}

Suppose that we have experimental data $\{Y_i,D_i,X_i\}_{i=1}^n$, where $X_i \in \mathbb{R}^{d_x}$ is a vector of observable pre-treatment covariates, $D_i \in \{0,1\}$ is a binary indicator of the treatment, and $Y_i$ is a binary outcome. Then $Y_i$ satisfies
\begin{equation}
    Y_i \ = \ \mu(D_i,X_i) + U_i, \ \ \ \ E[U_i|D_i,X_i]=0. \label{model}
\end{equation}
This implies $Y_i|D_i,X_i \sim Ber \left( \mu(D_i,X_i) \right)$. Under unconfoundedness assumption, we can consider $\mu(1,x) - \mu(0,x)$ as the CATE.

For simplicity, we assume that $D_i$ and $X_i$ are deterministic, $D_i=1$ for $i = 1, \ldots, n_1$, and $D_i = 0$ for $i = n_1+1, \ldots, n_1+n_0$, where $n_1+n_0 = n$. We define $Y_{1,i} \equiv Y_i$ for $i = 1, \ldots, n_1$, $Y_{0,i} \equiv Y_{n_1 + i}$ for $i = 1, \ldots, n_0$, $X_{1,i} \equiv X_i$ for $i = 1, \ldots, n_1$, and $X_{0,i} \equiv X_{n_1 + i}$ for $i = 1, \ldots, n_0$. Letting $p_{1,i} \equiv \mu(1,X_{1,i})$ for $i=1,\ldots,n_1$ and $p_{0,i} \equiv \mu(0,X_{0,i})$ for $i=1,\ldots,n_0$, then $Y_{1,i} \sim Ber(p_{1,i})$ for $i=1,\ldots,n_1$ and $Y_{0,i} \sim Ber(p_{0,i})$ for $i=1,\ldots,n_0$. Additionally, we define $p_{1,0} \equiv \mu(1,0)$ and $p_{0,0} \equiv \mu(0,0)$. Then, the distributions of $\bm{Y}_1 \equiv (Y_{1,1},\ldots,Y_{1,n_1})'$ and $\bm{Y}_0 \equiv (Y_{0,1},\ldots,Y_{0,n_0})'$ are determined by the parameters $\bm{p}_1 \equiv (p_{1,0}, \ldots, p_{1,n_1})'$ and $\bm{p}_0 \equiv (p_{0,0}, \ldots, p_{0,n_0})'$.

Throughout this paper, we consider the planner seeking to determine whether to treat individuals with $X_i = x_0$ based on the data $\mathbf{D} \equiv (\bm{Y}_1, \bm{Y}_0)$ given that the parameters $\bm{p}_1 \equiv (p_{1,0}, \ldots, p_{1,n_1})'$ and $\bm{p}_0 \equiv (p_{0,0}, \ldots, p_{0,n_0})'$ are unknown, where $x_0$ is a specific value of the covariate vector. The value $x_0$ does not have to be included in the support of the covariate distribution in data. Without loss of generality, we assume that $x_0 = 0$.

Following \citet{manski2004statistical, Manski2007}, \citet{stoye2009minimax, stoye2012minimax}, \citet{tetenov2012statistical}, \citet{ishihara2021evidence}, and \citet{yata2021optimal}, we focus on the minimax regret criterion to solve the decision problem. To this end, we assume that the true parameters $\bm{p} \equiv (\bm{p}_1, \bm{p}_0)$ are known to belong to the parameter space $\mathcal{P} \equiv \mathcal{P}_1 \times \mathcal{P}_0$. We impose the following restrictions on the parameter space:
\begin{eqnarray}
\mathcal{P}_1 &\equiv & \{\bm{p}_1 \in [0,1]^{n_1+1}: |p_{1,i} - p_{1,j}| \leq C \| X_{1,i} - X_{1,j} \| \}, \nonumber \\
\mathcal{P}_0 &\equiv & \{\bm{p}_0 \in [0,1]^{n_0+1}: |p_{0,i} - p_{0,j}| \leq C \| X_{0,i} - X_{0,j} \| \}, \nonumber
\end{eqnarray}
where $X_{1,0} = X_{0,0} \equiv 0$. This assumption implies that $\mu(1,x)$ and $\mu(0,x)$ are Lipschitz functions with the Lipschitz constant $C>0$.

\section{Main Results}
\subsection{Welfare and regret}

Given a treatment choice action $\delta \in [0,1]$, we define welfare attained by $\delta$ as
\begin{equation}
W(\delta) \ \equiv \ (p_{1,0} - p_{0,0}) \cdot \delta + p_{0,0}. \label{welfare}
\end{equation}
An optimal treatment choice action given knowledge of $\bm{p}$ is
\begin{equation}
\delta^{\ast} \ \equiv \ 1 \left\{ p_{1,0} \geq p_{0,0} \right\}.
\end{equation}
Then $W(\delta^{\ast})$ is the optimal welfare that would be achievable if we knew $\bm{p}$.

Let $\hat{\delta}(\mathbf{D}) \in \{0,1\}$ be a statistical treatment rule that maps data $\mathbf{D}$ to the binary treatment choice decision. The welfare regret of $\hat{\delta}(\mathbf{D})$ is defined as
\begin{eqnarray}
R(\bm{p}, \hat{\delta}) &\equiv & E_{\bm{p}} \left[ W(\delta^{\ast}) - W(\hat{\delta}(\mathbf{D})) \right] \nonumber \\
&=& (p_{1,0} - p_{0,0}) \cdot \left( \delta^{\ast} - E_{\bm{p}} [  \hat{\delta}(\mathbf{D}) ] \right), \label{regret}
\end{eqnarray}
where $E_{\bm{p}}(\cdot)$ is the expectation with respect to the sampling distribution of $\mathbf{D}$ given the parameters $\bm{p}$.

This study focuses on the following statistical treatment rule.

\begin{Assumption}\label{ass:class}
We consider the following class of statistical treatment rules $\mathcal{D}$:
\begin{eqnarray}
\mathcal{D} \equiv \left\{ \hat{\delta}_{\theta}(\mathbf{D}) \equiv 1 \left\{ \frac{\sum_{i=1}^{n_1} K(X_{1,i}/\theta) Y_{1,i}}{\sum_{i=1}^{n_1} K(X_{1,i}/\theta)} - \frac{\sum_{i=1}^{n_0} K(X_{0,i}/\theta) Y_{0,i}}{\sum_{i=1}^{n_0} K(X_{0,i}/\theta)}  \geq 0 \right\} \, : \, \theta > 0 \right\}, \label{ker_class}
\end{eqnarray}
where $K: \mathbb{R}^{d_x} \to \mathbb{R}_{+}$ denotes the kernel function.
\end{Assumption}

Assumption \ref{ass:class} implies that we focus on non-randomized statistical treatment rules that plug in the fitted values based on nonparametric kernel regression. In addition, we assume that the kernel function takes non-negative values. Our results are dependent on this condition. In (\ref{ker_class}), $\theta$ is the bandwidth of the kernel regression estimator.

The minimax regret criterion selects a statistical treatment rule that minimizes the maximum regret:
\begin{equation}
\hat{\delta}_{\mathcal{D}}^{\ast} \ = \ \text{arg} \min_{\hat{\delta} \in \mathcal{D}} \max_{\bm{p} \in \mathcal{P}} R(\bm{p}, \hat{\delta}). \nonumber
\end{equation}
Since a statistical treatment rule $\hat{\delta} \in \mathcal{D}$ is characterized by bandwidth $\theta$, the optimal bandwidth can be calculated as follows:
\begin{equation}
\theta^{\ast} \ = \ \text{arg} \min_{\theta} \max_{\bm{p} \in \mathcal{P}} R(\bm{p}, \hat{\delta}_{\theta}), \label{optimal_bw}
\end{equation}
where $\hat{\delta}_{\theta}$ is as defined in (\ref{ker_class}). In the next subsection, we describe the calculation of the optimal bandwidth $\theta^{\ast}$.

\subsection{Minimax regret rule}

The following theorem implies that we can calculate $\max_{\bm{p} \in \mathcal{P}} R(\bm{p},\hat{\delta}_{\theta})$ by optimizing two parameters.

\begin{Theorem}\label{thm:main}
Under Assumption \ref{ass:class}, for any $\hat{\delta}_{\theta} \in \mathcal{D}$ we obtain
\begin{eqnarray}
\max_{\bm{p} \in \mathcal{P}} R(\bm{p},\hat{\delta}_{\theta}) &=& \max \left[ \max_{p_{1,0} > p_{0,0}} \left\{ (p_{1,0}-p_{0,0}) \cdot \left( 1 - E_{\tilde{\bm{p}}^{-}(p_{1,0}, p_{0,0})} [\hat{\delta}_{\theta}(\mathbf{D}) ] \right) \right\} ,  \right. \nonumber \\
& & \hspace{1.3in} \left. \max_{p_{1,0} < p_{0,0}} \left\{ (p_{0,0}-p_{1,0}) \cdot E_{\tilde{\bm{p}}^{+}(p_{1,0}, p_{0,0})} [  \hat{\delta}_{\theta}(\mathbf{D}) ] \right\} \right], \label{Thm_main}
\end{eqnarray}
where $\tilde{\bm{p}}^{-}(p_{1,0}, p_{0,0}) \equiv \left( \tilde{\bm{p}}_1^{-}(p_{1,0}), \tilde{\bm{p}}_0^{+}(p_{0,0}) \right)$, $\tilde{\bm{p}}^{+}(p_{1,0}, p_{0,0}) \equiv \left( \tilde{\bm{p}}_1^{+}(p_{1,0}), \tilde{\bm{p}}_0^{-}(p_{0,0}) \right)$,
\begin{eqnarray*}
\tilde{\bm{p}}_1^{-}(p) &\equiv & \left( \tilde{p}_{1,0}^{-}(p), \tilde{p}_{1,1}^{-}(p), \ldots, \tilde{p}_{1,n_1}^{-}(p) \right)' \\
&\equiv & \left( p, \max \{ p-C\| X_{1,1} \|, 0 \}, \ldots, \max \{ p-C\| X_{1,n_1} \|, 0 \} \right)', \\
\tilde{\bm{p}}_1^{+}(p) &\equiv & \left( \tilde{p}_{1,0}^{+}(p), \tilde{p}_{1,1}^{+}(p), \ldots, \tilde{p}_{1,n_1}^{+}(p) \right)' \\
&\equiv & \left( p, \min \{ p+C\| X_{1,1} \|, 1 \}, \ldots, \min \{ p+C\| X_{1,n_1} \|, 1 \} \right)', \\
\tilde{\bm{p}}_0^{-}(p) &\equiv & \left( \tilde{p}_{0,0}^{-}(p), \tilde{p}_{0,1}^{-}(p), \ldots, \tilde{p}_{0,n_0}^{-}(p) \right)' \\
&\equiv & \ \left( p, \max \{ p-C\| X_{0,1} \|, 0 \}, \ldots, \max \{ p-C\| X_{0,n_0} \|, 0 \} \right)', \\
\tilde{\bm{p}}_0^{+}(p) &\equiv & \left( \tilde{p}_{0,0}^{+}(p), \tilde{p}_{0,1}^{+}(p), \ldots, \tilde{p}_{0,n_0}^{+}(p) \right)' \\
&\equiv & \left( p, \min \{ p+C\| X_{0,1} \|, 1 \}, \ldots, \min \{ p+C\| X_{0,n_0} \|, 1 \} \right)'.
\end{eqnarray*}
\end{Theorem}

The proof of Theorem \ref{thm:main} provides parameters that maximize the regret of $\hat{\delta}_{\theta} \in \mathcal{D}$. If $p_{1,0} > p_{0,0}$, then the regret of $\hat{\delta}_{\theta}$ is maximized at
\begin{eqnarray*}
& & \bm{p}_1 = \left( p_{1,0}, \max \{ p_{1,0}-C\| X_{1,1} \|, 0 \}, \ldots , \max \{ p_{1,0}-C\| X_{1,n_1} \|, 0 \} \right)' \ \text{and} \\
& & \bm{p}_0 = \left( p_{0,0}, \min \{ p_{0,0}+C\| X_{0,1} \|, 1 \}, \ldots , \min \{ p_{0,0}+C\| X_{0,n_0} \|, 1 \} \right)'.
\end{eqnarray*}
Similarly, if $p_{1,0} < p_{0,0}$, then the regret of $\hat{\delta}_{\theta}$ is maximized at
\begin{eqnarray*}
& & \bm{p}_1 = \left( p_{1,0}, \min \{ p_{1,0}+C\| X_{1,1} \|, 1 \}, \ldots , \min \{ p_{1,0}+C\| X_{1, n_1} \|, 1 \} \right)' \ \text{and} \\
& & \bm{p}_0 = \left( p_{0,0}, \max \{ p_{0,0}-C\| X_{0,1} \|, 0 \}, \ldots , \max \{ p_{0,0}-C\| X_{0,n_0} \|, 0 \} \right)'.
\end{eqnarray*}
Using these results, we can calculate the maximum regret when $p_{1,0}$ and $p_{0,0}$ are fixed. Hence, we can compute $\max_{\bm{p} \in \mathcal{P}} R(\bm{p},\hat{\delta}_{\theta})$ by optimizing $p_{1,0}$ and $p_{0,0}$.

In view of Theorem \ref{thm:main}, we can compute the optimal bandwidth $\theta^{\ast}$ using the following algorithm. \vspace{0.05in}
\begin{itemize}
\item[\textbf{1.}] Fix $\hat{\delta}_{\theta} \in \mathcal{D}$ and generate uniform distributed random variables $\{U_{1,i,s}\}_{i=1, \ldots, n_1, \, s = 1, \ldots, S}$ and $\{U_{0,i,s}\}_{i=1, \ldots, n_0, \, s = 1, \ldots, S}$, where $S$ is a large number. Define the following variables:
\begin{eqnarray*}
& & \tilde{Y}_{1,i,s}^{-}(p) \equiv 1 \left\{ U_{1,i,s} \leq \tilde{p}_{1,i}^{-}(p) \right\}, \ \ \tilde{Y}_{1,i,s}^{+}(p) \equiv 1 \left\{ U_{1,i,s} \leq \tilde{p}_{1,i}^{+}(p) \right\}, \\
& & \tilde{Y}_{0,i,s}^{-}(p) \equiv 1 \left\{ U_{0,i,s} \leq \tilde{p}_{0,i}^{-}(p) \right\}, \ \ \tilde{Y}_{0,i,s}^{+}(p) \equiv 1 \left\{ U_{0,i,s} \leq \tilde{p}_{0,i}^{+}(p) \right\}, \\
& & \tilde{\bm{Y}}^{-}_{1,s}(p) \equiv \left( \tilde{Y}_{1,1,s}^{-}(p), \ldots , \tilde{Y}_{1,n_1,s}^{-}(p) \right),  \ \ \tilde{\bm{Y}}^{+}_{1,s}(p) \equiv \left( \tilde{Y}_{1,1,s}^{+}(p), \ldots , \tilde{Y}_{1,n_1,s}^{+}(p) \right), \\
& & \tilde{\bm{Y}}^{-}_{0,s}(p) \equiv \left( \tilde{Y}_{0,1,s}^{-}(p), \ldots , \tilde{Y}_{0,n_0,s}^{-}(p) \right), \ \ \tilde{\bm{Y}}^{+}_{0,s}(p) \equiv \left( \tilde{Y}_{0,1,s}^{+}(p), \ldots , \tilde{Y}_{0,n_0,s}^{+}(p) \right).
\end{eqnarray*}
\item[\textbf{2.}] We approximate $E_{\tilde{\bm{p}}^{-}(p_{1,0}, p_{0,0})} [ \hat{\delta}_{\theta}(\mathbf{D}) ]$ and $E_{\tilde{\bm{p}}^{+}(p_{1,0}, p_{0,0})} [ \hat{\delta}_{\theta}(\mathbf{D}) ]$ by $\pi^{-}_{\theta}(p_{1,0},p_{0,0})$ and $\pi^{+}_{\theta}(p_{1,0},p_{0,0})$, where
\begin{eqnarray*}
\pi^{-}_{\theta}(p_{1,0},p_{0,0}) &\equiv & \frac{1}{S} \sum_{s=1}^S \hat{\delta}_{\theta}( \tilde{\bm{Y}}^{-}_{1,s}(p_{1,0}), \tilde{\bm{Y}}^{+}_{0,s}(p_{0,0})), \\
\pi^{+}_{\theta}(p_{1,0},p_{0,0}) &\equiv & \frac{1}{S} \sum_{s=1}^S \hat{\delta}_{\theta}( \tilde{\bm{Y}}^{+}_{1,s}(p_{1,0}), \tilde{\bm{Y}}^{-}_{0,s}(p_{0,0})).
\end{eqnarray*}
\item[\textbf{3.}] Calculate the maximum regret $\overline{R}(\theta)$, where
\begin{eqnarray*}
\overline{R}(\theta) & \equiv & \max \left[ \max_{p_{1,0}>p_{0,0}} \left\{ (p_{1,0}-p_{0,0}) \left( 1 - \pi^{-}_{\theta}(p_{1,0},p_{0,0}) \right) \right\}, \right. \\
& & \hspace{1.3in} \left. \max_{p_{1,0}<p_{0,0}} \left\{ (p_{0,0}-p_{1,0}) \pi^{+}_{\theta}(p_{1,0},p_{0,0}) \right\} \right].
\end{eqnarray*}
\item[\textbf{4.}] Minimizes $\overline{R}(\theta)$ and obtain the optimal bandwidth $\theta^{\ast}$.
\end{itemize}

This algorithm is advantageous from a computational perspective. Computation of the exact minimax regret rule is often challenging in the context of statistical treatment choices. In situations where the sample size is large, calculating the maximum regret necessitates addressing a substantial-dimensional optimization problem. However, using Theorem \ref{thm:main}, it is possible to calculate the maximum regret with greater ease. In the next section, we compute $\theta^{\ast}$ by using this algorithm.

\begin{Remark}
Similar to \cite{ishihara2021evidence}, when the Lipschitz constant $C$ is unknown, we do not know how to select $C$ in a theoretically justifiable data-driven manner. However, \cite{ishihara2021evidence} and \cite{yata2021optimal} propose some practical choices for $C$. In the empirical application, \cite{ishihara2021evidence} perform leave-one-out
cross-validation to select $C$. \cite{yata2021optimal} estimates a lower bound on $C$ by using the derivative of the conditional mean function. In our setting, we can also apply both of these methods.
\end{Remark}

\begin{Remark}
The maximum regret does not depend on the bandwidth when $C$ is sufficiently large. If $X_{1,0}, X_{1,1}, \ldots, X_{1,n_1}$ and $X_{0,0}, X_{0,1}, \ldots, X_{0,n_0}$ differ, the parameter space $\mathcal{P}$ becomes $[0,1]^{n_1+1} \times [0,1]^{n_0+1}$ when $C$ is large enough. From Theorem \ref{thm:main}, if $p_{1,0} > p_{0,0}$, then the regret of $\hat{\delta}_{\theta}$ is maximized at
\[
\bm{p}_1 = (p_{1,0},0,\ldots,0)' \ \ \text{and} \ \ \bm{p}_0 = (p_{0,0},1,\ldots,1)'.
\]
Similarly, if $p_{1,0} < p_{0,0}$, then the regret of $\hat{\delta}_{\theta}$ is maximized at $\bm{p}_1 = (p_{1,0},1,\ldots,1)'$ and $\bm{p}_0 = (p_{0,0},0,\ldots,0)'$. These results imply that
\[
\max_{\bm{p} \in \mathcal{P}} R(\bm{p},\hat{\delta}_{\theta}) \ = \ \max \left\{ \max_{p_{1,0} > p_{0,0}} (p_{1,0} - p_{0,0}), \max_{p_{1,0} < p_{0,0}} (p_{0,0} - p_{1,0}) \right\} \ = \ 1.
\]
Hence, the maximum regret does not depend on the bandwidth $\theta$.
\end{Remark}

\section{Numerical Examples}

In this section, we perform a numerical analysis to compare the optimal bandwidth explained in the previous section with that of the normally distributed outcomes. Throughout this section, we set the pre-treatment covariates as equidistant grid points on $[-1,1]$:
\begin{eqnarray*}
X_{1,i} &=& -1 + \frac{2(i-1)}{n_1-1} \ \ \ \text{for $i = 1, \ldots, n_1$,} \\
X_{0,i} &=& -1 + \frac{2(i-1)}{n_0-1} \ \ \ \text{for $i = 1, \ldots, n_0$,}
\end{eqnarray*}
where we set $n_1=n_0=n/2$. We consider the kernel regression class $\mathcal{D}$ defined as (\ref{ker_class}), where we use the Gaussian kernel as the kernel function. We calculate the optimal bandwidth $\theta^{\ast}$ by using the proposed method.

We compare our method with bandwidth selection, which minimizes the maximum regret when the outcome variables are normally distributed. Suppose that $Y_{1,i} \sim N(p_{1,i}, 0.5^2)$ and $Y_{0,i} \sim N(p_{0,i}, 0.5^2)$. Then, we set the parameter space as $\mathcal{P}^{N} \equiv \mathcal{P}^N_1 \times \mathcal{P}_0^N$, where
\begin{eqnarray}
\mathcal{P}_1^N &\equiv & \{\bm{p}_1 \in \mathbb{R}^{n_1+1}: |p_{1,i} - p_{1,j}| \leq C \| X_{1,i} - X_{1,j} \| \}, \nonumber \\
\mathcal{P}_0^N &\equiv & \{\bm{p}_0 \in \mathbb{R}^{n_0+1}: |p_{0,i} - p_{0,j}| \leq C \| X_{0,i} - X_{0,j} \| \}. \nonumber
\end{eqnarray}
Define $w_{1,i,\theta} \equiv K(X_{1,i}/\theta) / \sum_{i=1}^{n_1} K(X_{1,i}/\theta)$ and $w_{0,i,\theta} \equiv K(X_{0,i}/\theta) / \sum_{i=1}^{n_0} K(X_{0,i}/\theta)$. Using the argument of \citet{ishihara2021evidence}, the maximum regret can be expressed as follows:
\begin{equation}
\max_{\bm{p} \in \mathcal{P}^N} R(\bm{p},\hat{\delta}_{\theta}) \ = \ s(\theta) \eta \left( \frac{b(\theta)}{s(\theta)} \right),\label{max_reg_normal}
\end{equation}
where $\Phi(\cdot)$ is a distribution function of $N(0,1)$, $\eta(a) \equiv \max_{t>0} \{ t \cdot \Phi(-t+a) \}$, $s(\theta) \equiv 0.5 \sqrt{\sum_{i=1}^{n_1} w_{1,i,\theta}^2 + \sum_{i=1}^{n_0} w_{0,i,\theta}^2}$, and
\[
b(\theta) \equiv C \left( \sum_{i=1}^{n_1} w_{1,i,\theta} \|X_{1,i}\| + \sum_{i=1}^{n_0} w_{0,i,\theta} \|X_{0,i}\| \right).
\]
Appendix 2 provides the details of the derivation. Using this result, we calculate the bandwidth that minimizes the maximum regret for the normally distributed outcomes.

\begin{table}[htb]
\caption{Bandwidth choices}
\begin{center} 
   \begin{tabular}{c c c c c c c c c} \hline
      & \multicolumn{2}{c}{$n=10$} & \multicolumn{2}{c}{$n=50$} & \multicolumn{2}{c}{$n=100$} & \multicolumn{2}{c}{$n=200$} \\
      & binary & normal & binary & normal & binary & normal & binary & normal \\ \hline \hline 
      $C=0.1$ & 0.12--0.40 & 0.64 & 0.36 & 0.34 & 0.28 & 0.26 & 0.22 & 0.21  \\ 
      $C=0.2$ & 0.12--0.39 & 0.37 & 0.19 & 0.21 & 0.17 & 0.16 & 0.13 & 0.13  \\ 
      $C=0.3$ & 0.12--0.39 & 0.30 & 0.09 & 0.16 & 0.11 & 0.13 & 0.09 & 0.10  \\ \hline     
  \end{tabular}
\end{center}
\end{table}

Table 1 lists the optimal bandwidth choices for binary and normal cases. In the binary case, we calculate the optimal bandwidth by using the algorithm in Section 3.2. In the normal case, we calculate the bandwidth that minimizes (\ref{max_reg_normal}), that is, the bandwidth choice proposed by \cite{ishihara2021evidence}.

In many cases, the optimal bandwidth decreases as $C$ decreases or $n$ increases. When $n$ is large, the optimal bandwidth for the binary outcomes is close to that for normally distributed outcomes. This phenomenon arises due to the asymptotic normality of the kernel regression estimator. In the binary case, the maximum regret (\ref{Thm_main}) is a step function with respect to $\theta$ and approaches a continuous function as $n$ increases. Hence, Table 1 shows the optimal bandwidth range when $n=10$. When $n$ is small, the optimal bandwidth for the binary outcomes can be significantly different from that for normally distributed outcomes.

\section{Conclusion}

This study investigated whether to treat individuals based on observed covariates. The standard approach to this problem is to use a plug-in rule that determines the treatment based on the sign of an estimate of the CATE. We focused on statistical treatment rules based on nonparametric kernel regression. In situations in which the outcome variables are normally distributed, \citet{ishihara2021evidence} showed that the maximum regret can be calculated easily. This study demonstrated that the maximum regret can be calculated by optimizing two parameters even when the outcome variables are binary. Using these results, we  proposed an optimal bandwidth selection method for the binary outcomes. In addition, we performed a numerical analysis to compare the optimal bandwidth choices for binary and normally distributed outcomes.

\clearpage
\renewcommand{\theequation}{A.\arabic{equation}}
\setcounter{equation}{0}
\section*{Appendix 1: Proofs of Theorem \ref{thm:main} and Lemma \ref{lem:s-d}}

\begin{proof}[Proof of Theorem \ref{thm:main}]
If $\bm{p}$ satisfies $p_{1,0} > p_{0,0}$, then we have
\begin{eqnarray*}
R(\bm{p},\hat{\delta}_{\theta}) &=& (p_{1,0}-p_{0,0}) \cdot \left( 1 - E_{\bm{p}} [\hat{\delta}_{\theta}(\mathbf{D}) ] \right).
\end{eqnarray*}
For any $p \in [0,1]$, $\tilde{\bm{p}}_1^{-}(p)$ and $\tilde{\bm{p}}_0^{+}(p)$ are contained in $\mathcal{P}_1$ and $\mathcal{P}_0$, respectively. Additionally, we have $p_{1,i} \geq \tilde{p}_{1,i}^{-}(p_{1,0})$ and $p_{0,i} \leq \tilde{p}_{0,i}^{+}(p_{0,0})$ for all $i$. Because $Ber(p)$ has first-order stochastic dominance over $Ber(\tilde{p})$ for $p \geq \tilde{p}$, it follows from Assumption \ref{ass:class} and Lemma \ref{lem:s-d} that
\begin{equation*}
E_{\bm{p}} [\hat{\delta}_{\theta}(\mathbf{D}) ] \ \geq \ E_{\tilde{\bm{p}}^{-}(p_{1,0}, p_{0,0})} [\hat{\delta}_{\theta}(\mathbf{D}) ].
\end{equation*}
Hence, we obtain
\begin{equation*}
\max_{p_{1,0} > p_{0,0}} R(\bm{p},\hat{\delta}_{\theta}) \ = \ \max_{p_{1,0} > p_{0,0}} \left\{ (p_{1,0}-p_{0,0}) \cdot \left( 1 - E_{\tilde{\bm{p}}^{-}(p_{1,0}, p_{0,0})} [\hat{\delta}_{\theta}(\mathbf{D}) ] \right) \right\}.
\end{equation*}
Similarly, we obtain
\begin{equation*}
\max_{p_{1,0} < p_{0,0}} R(\bm{p},\hat{\delta}_{\theta}) \ = \ \max_{p_{1,0} < p_{0,0}} \left\{ (p_{0,0}-p_{1,0}) \cdot E_{\tilde{\bm{p}}^{+}(p_{1,0}, p_{0,0})} [ \hat{\delta}_{\theta}(\mathbf{D}) ] \right\}.
\end{equation*}
As $p_{1,0}=p_{0,0}$ implies $R(\bm{p},\hat{\delta}_{\theta}) = 0$, we obtain (\ref{Thm_main}).
\end{proof}

\begin{Lemma}\label{lem:s-d}
Suppose that $W_1, W_2, \ldots, W_n$, and $\tilde{W}_1$ are independent random variables and $\tilde{W}_1$ has first-order stochastic dominance over $W_1$. If $g(w_1, \ldots, w_n)$ is increasing in $w_1$, then we have $E\left[ g(W_1, W_2, \ldots , W_n) \right] \leq E\left[ g(\tilde{W}_1, W_2, \ldots , W_n) \right]$.
\end{Lemma}

\begin{proof}
Let $F_i$ be a distribution function of $W_i$ and $\tilde{F}_1$ be a distribution function of $\tilde{W}_1$. Because $g(w_1, \ldots, w_n)$ is increasing in $w_1$, we have $\int g(w_1, \ldots, w_n) dF_1(w_1) \leq \int g(w_1, \ldots, w_n) d\tilde{F}_1(w_1)$. Hence, we obtain 
\begin{eqnarray*}
E\left[ g(W_1, W_2, \ldots , W_n) \right] &=& \int \cdots \int g(w_1, \ldots, w_n) dF_1(w_1) \cdots dF_n(w_n) \\
& \leq & \int \cdots \int g(w_1, \ldots, w_n) d\tilde{F}_1(w_1) \cdots dF_n(w_n) \\
&=& E\left[ g(\tilde{W}_1, W_2, \ldots , W_n) \right].
\end{eqnarray*}
\end{proof}

\section*{Appendix 2: Derivation of (\ref{max_reg_normal})}

In this section, we present the derivation of (\ref{max_reg_normal}). Suppose that $Y_{1,i} \sim N(p_{1,i}, 0.5^2)$ and $Y_{0,i} \sim N(p_{0,i}, 0.5^2)$. Then, we set the parameter space as $\mathcal{P}^{N} \equiv \mathcal{P}^N_1 \times \mathcal{P}_0^N$, where
\begin{eqnarray}
\mathcal{P}_1^N &\equiv & \{\bm{p}_1 \in \mathbb{R}^{n_1+1}: |p_{1,i} - p_{1,j}| \leq C \| X_{1,i} - X_{1,j} \| \}, \nonumber \\
\mathcal{P}_0^N &\equiv & \{\bm{p}_0 \in \mathbb{R}^{n_0+1}: |p_{0,i} - p_{0,j}| \leq C \| X_{0,i} - X_{0,j} \| \}. \nonumber
\end{eqnarray}
Define $w_{1,i,\theta} \equiv K(X_{1,i}/\theta) / \sum_{i=1}^{n_1} K(X_{1,i}/\theta)$ and $w_{0,i,\theta} \equiv K(X_{0,i}/\theta) / \sum_{i=1}^{n_0} K(X_{0,i}/\theta)$. Then the regret can be written as
\begin{eqnarray*}
R(\bm{p}, \hat{\delta}_{\theta}) &=& (p_{1,0}-p_{0,0}) \cdot \left\{ 1\{p_{1,0} \geq p_{0,0} \} - \Phi \left( \frac{\sum_{i=1}^{n_1} w_{1,i,\theta} p_{1,i} - \sum_{i=1}^{n_0} w_{0,i,\theta} p_{0,i} }{s(\theta)} \right) \right\} \\
&=& |p_{1,0}-p_{0,0}| \cdot \Phi \left( - sgn(p_{1,0}-p_{0,0}) \cdot \frac{\sum_{i=1}^{n_1} w_{1,i,\theta} p_{1,i} - \sum_{i=1}^{n_0} w_{0,i,\theta} p_{0,i} }{s(\theta)} \right).
\end{eqnarray*}
From the symmetry of $\mathcal{P}^N$, we have
\begin{eqnarray*}
& & \max_{\bm{p} \in \mathcal{P}^N} R(\bm{p}, \hat{\delta}_{\theta}) \\
&=& \max_{t > 0} \max_{\substack{\bm{p} \in \mathcal{P}^N \\ p_{1,0} - p_{0,0} = t}} t \cdot \Phi \left( - \frac{t + \sum_{i=1}^{n_1} w_{1,i,\theta} (p_{1,i}-p_{1,0}) - \sum_{i=1}^{n_0} w_{0,i,\theta} (p_{0,i} - p_{0,0})}{s(\theta)} \right).
\end{eqnarray*}
From the definitions of $\mathcal{P}_1^N$ and $\mathcal{P}_0^N$, we obtain
\begin{eqnarray*}
& & - \min_{\substack{\bm{p} \in \mathcal{P}^N \\ p_{1,0} - p_{0,0} = t}} \left\{ \sum_{i=1}^{n_1} w_{1,i,\theta} (p_{1,i}-p_{1,0}) - \sum_{i=1}^{n_0} w_{0,i,\theta} (p_{0,i} - p_{0,0}) \right\} \\
&=& \max_{\substack{\bm{p}_1 \in \mathcal{P}_1^N \\ p_{1,0}=0}} \left\{ \sum_{i=1}^{n_1} w_{1,i,\theta} p_{1,i} \right\} + \max_{\substack{\bm{p}_0 \in \mathcal{P}_0^N \\ p_{0,0}=0}} \left\{ \sum_{i=1}^{n_0} w_{0,i,\theta} p_{0,i} \right\} \ = \ b(\theta).
\end{eqnarray*}
Hence, we obtain (\ref{max_reg_normal}). This result is essentially the same as that in \citet{ishihara2021evidence}.

\bigskip

\clearpage
\bibliographystyle{ecta}
\bibliography{meta-analysis}

\end{document}